\documentclass[runningheads]{llncs}

\usepackage{graphicx}
\usepackage{tikz}
\usepackage{graphicx}
\usepackage{amsmath}
\usepackage{color}
\usepackage{xcolor}

\usepackage{mathtools}
\usepackage{framed}
\usepackage{amsfonts}

\begin{document}

\title{Fixed-price Diffusion Mechanism Design}


%

%
%

\author{Tianyi Zhang \and
Dengji Zhao \and
Wen Zhang \and
Xuming He}
\authorrunning{T. Zhang et al.}
%
\institute{School of Information Science and Technology, \\ShanghaiTech University, Shanghai, China\\
\email{\{zhangty,zhaodj,zhangwen,hexm\}@shanghaitech.edu.cn}}

%
\maketitle              
%

\begin{abstract}
We consider a fixed-price mechanism design setting where a seller sells one item via a social network, but the seller can only directly communicate with her neighbours initially. Each other node in the network is a potential buyer with a valuation derived from a common distribution. With a standard fixed-price mechanism, the seller can only sell the item among her neighbours. To improve her revenue, she needs more buyers to join in the sale. To achieve this, we propose the very first fixed-price mechanism to incentivize the seller's neighbours to inform their neighbours about the sale and to eventually inform all buyers in the network to improve seller's revenue. Compared with the existing mechanisms for the same purpose, our mechanism does not require the buyers to reveal their valuations and it is computationally easy. More importantly, it guarantees that the improved revenue is at least 1/2 of the optimal.


\end{abstract}
\section{Introduction}
Social networks play a vital role in our daily activities, especially in electronic commerce. How to utilize social networks for promotions is widely studied by both researchers and practitioners~\cite{Jackson2005TheEO}. Mechanism design on social networks, aiming to establish specific policies satisfying some desirable properties, has received much attention from the fields of economics, computer science and artificial intelligence~\cite{Domingos:2001:MNV:502512.502525}.

The fixed-price mechanism has been the mostly applied trading rules for a seller to sell products in our daily life. The seller sets a fixed price before the buyers buy the product and a buyer would buy it if the price is not over her valuation for the product. Due to its simplicity, there does not exist much research on it~\cite{mathews2004impact}.


However, determining an optimal price under a fixed-price mechanism is still challenging. If the seller aims to maximize her revenue, she could find the optimal price given the number of buyers in the market and their valuation distributions. We will show that the number of buyers can significantly influence the seller's revenue. Specifically, even under fixed-price mechanisms, the more buyers participating in the mechanism, the higher revenue for the seller. Therefore, the seller is incentivized to promote the sale to attract more buyers to increase her revenue, even if there is only one item to sell.     


This paper aims to design a novel fixed-price mechanism to help the seller to attract more buyers via a social network. We assume that the seller is located on the social network and she can sell the item with a fixed-price mechanism to her neighbours without promotion. We want the seller's neighbours to help her attract more buyers from their neighbours. However, her neighbours would not do so without giving them a proper incentive (why should they bring competitors for the item). The goal of our mechanism is to design such an incentive for the neighbours to invite more buyers.


A similar problem has been studied recently in \cite{article,zhao2018selling} where they proposed an auction with a dedicated payment scheme to attract more buyers via a social network. The payment requires the whole structure of the network and also the valuations of all buyers, while our proposed mechanism is fixed-priced and computationally easy with a guaranteed revenue improvement for the seller. 

Our mechanism is also inspired by the mechanism with rewards used by the MIT Red Balloon Challenge Team to win the 2009 DARPA Networks Competition sponsored by a research organization of the United States Department of Defense~\cite{pickard2011time}. The challenge is to locate 10 balloons located at undisclosed locations across the U.S.. With the dedicated rewarding scheme, the MIT team only began with four initial participants but eventually attracted over 5,000 participates and won the challenge. We extended this rewarding scheme into our mechanism to incentivize buyers to invite their neighbours.

As closely related work, fixed-price mechanisms have also been investigated in other mechanism design settings \cite{wang2008auction,tran2017group}. Goldberg \textit{et al.} used an optimal fixed pricing as a benchmark for evaluating the expected revenue of single-round and sealed-bid auctions~\cite{goldberg2006competitive}. Emek \textit{et al.} analyzed the properties of the well-known geometric reward mechanism that incentivizes the participation for multi-level marketing~\cite{emek2011mechanisms}. McAfee showed under the condition that the buyer's distribution dominates the seller's distribution for a bilateral trade, a fixed-price mechanism gets at least half of the maximum social welfare gain from the trade~\cite{mcafee2008gains}. Badanidiyuru \textit{et al.} presented a constant-competitive fixed-price mechanism for online procurement markets with sub-modular utility~\cite{badanidiyuru2012learning}.

However, to the best of our knowledge, no work has considered introducing information diffusion into fixed-price mechanisms. Thus, in this work, we propose a novel mechanism by integrating the merits of both fixed-price mechanism and information diffusion. Our fixed-price diffusion mechanism distinguishes itself with the following.

\begin{enumerate}
	\item It helps the seller to attract more buyers via social networks without paying them in advance. 
	\item It does not require buyers to reveal their valuations. The payments for buyers are relatively fixed and computationally easy. 
	\item The seller's revenue improvement is guaranteed and it is bounded above half of the optimal. 
\end{enumerate}

The rest of the paper is structured as follows. Section~\ref{The model} describes the model of our study. Section~\ref{without} investigates the traditional fixed-price mechanism. Section~\ref{FPDM} defines the proposed Fixed-price Diffusion Mechanism. Section~\ref{pro} studies the properties of the proposed mechanism and we conclude in Section~\ref{con}.

\section{The Model}
\label{The model}
We consider a market where a seller $s$ sells one item via a social network. We assume that the social network is a rooted tree, and the root is $s$ and there are $k$ other nodes in the tree denoted by $K=\{1,\cdots,k\}$. The $k$ other nodes are the potential buyers of the item sold by $s$. Each agent $i \in K $ has a set of neighbours denoted by $n_i \subseteq K \backslash \{i\}$, and she can only exchange information with her neighbours. Each buyer $i \in K$ has a valuation $v_i > 0$ for the item. We assume, as suggested by Myerson, buyers' valuations are independently derived from the uniform distribution on the interval $[0,1]$~\cite{myerson1981optimal}.

Without any advertising, the seller $s$ can only sell the item among her neighbours since she can not communicate with the rest buyers from the network directly. Fixed-price mechanisms have been widely applied for this kind of market due to their simplicity. However, determining an optimal fixed-price is challenging. Moreover, the more buyers in the market, the higher chance to sell it with a higher fixed-price. In order to attract more people to buy, the seller $s$ may pay third-party platforms to promote the sale. This is often very costly with no guarantee to improve the revenue. Therefore, in this paper, we consider another class of fixed-price mechanisms which can improve the seller's revenue by attracting more people to buy. It achieves this by incentivizing buyers to propagate the sale information to all their neighbours, instead of by the seller's advertising through other platforms. At the same time, it does not require buyers to report their valuations as standard fixed-price mechanisms would do. 

Since our mechanism does not require buyers to report their valuations, the action of a buyer is just diffusing the sale information to their neighbours. Let $n_i$ be the {\em type} of buyer $i \in K$, $(n_1,\cdots,n_k)$ be the type profile of all buyers and $(n_1,\cdots,n_{i-1},n_{i+1},\cdots,n_k)$ be the type profile of all buyers except $i$. Let $a_i'=n_i'$ or $nil$ be the {\em action} of $i$, where $n_i'$ represents the neighbours informed by $i$ and $nil$ represents that $i$ has not been informed or she does not want to participate in the sale. Furthermore, let $a'=(a_1',\cdots,a_k')$ be the action profile of all buyers and $a_{-i}'$ be the action profile of all buyers except $i$.
\begin{definition}
	Given the type profile $(n_1,\cdots,n_k)$ of all buyers, an action profile $a'$ is feasible if for all $i \in K$,
	\begin{itemize}
		\item $a_i'\ne nil$ if and only if $i$ can be informed the sale information from the seller following the action profile of $a_{-i}'$. 
		\item if $a_i'\ne nil$, then $n_i' \subseteq n_i$.
	\end{itemize}
\end{definition}
Let $\mathcal{F}((n_1,\cdots,n_k))$ be the set of all feasible action profiles of all buyers. The mechanism is defined only on feasible action profiles.

\begin{definition}
	A mechanism in the social network consists of an allocation policy $\pi=(\pi_i)_{i \in K}$ and a payment policy $p=(p_i)_{i \in K}$, where $\pi_i:\mathcal{F}((n_1,\cdots,n_k)) \rightarrow \left\{0,1\right\}$ and $p_i:\mathcal{F}((n_1,\cdots,n_k)) \rightarrow [0,1]$ are the allocation and payment functions for buyer $i$.
\end{definition}

In the above definition, given the action profile $a'$, $\pi_i(a')=1$ indicates that $i$ receives the item whereas $\pi_i(a')=0$ means that $i$ does not get the item. $p_i(a') > 0$ indicates $i$ pays $p_i(a')$ to the mechanism whereas $p_i(a') < 0$ means $i$ receives $|p_i(a')|$ from the mechanism.

Next, we introduce the desirable properties we want to design for the mechanism.
\begin{definition}
	An allocation $\pi$ is feasible if for all $a' \in \mathcal{F}((n_1,\cdots,n_k))$,
	\begin{itemize}
		\item for all $i \in K$, if $a_i'=nil$, then $\pi_i(a')=0$.
		\item $\sum_{i \in K}\pi_i(a') \leq 1$.
	\end{itemize}
\end{definition}
A feasible allocation cannot allocate the item to a buyer who does not join the sale and it should not allocate the item to multiple buyers. In the rest of the paper, only feasible allocations are considered.

Given a feasible action profile $a'$ and a mechanism $(\pi,p)$, the utility of $i$ is defined as
\begin{equation*}
u_i(a',(\pi,p))=\pi_i(a')v_i-p_i(a').
\end{equation*} 

A mechanism is individually rational if, for each buyer, her utility is always non-negative no matter how many neighbours she informs.

\begin{definition}
	A mechanism $(\pi,p)$ is individually rational (IR) if $u_i(a',(\pi,p))\geq 0$ for all $i \in K$, for all $a' \in \mathcal{F}((n_1,\cdots,n_k))$.
\end{definition}

In a standard fixed-price mechanism, a mechanism is incentive compatible if and only if for each buyer, reporting her truthful action is a dominant strategy~\cite{Nisan2009Algorithmic}. In our model, the action of each buyer is diffusing the sale information. Hence, incentive compatibility means that for each buyer, diffusing the sale information to all her neighbours is a dominant strategy.

\begin{definition}
	A mechanism $(\pi,p)$ is incentive compatible (IC) if $u_i(a',(\pi,p)) \geq u_i(a'',(\pi,p))$ for all $i \in K$, for all $a', a'' \in \mathcal{F}((n_1,\cdots,n_k))$ such that $a_i'=n_i$.
\end{definition}

In this paper, we design a fixed-price mechanism that is IC and IR, which encourages the buyers to diffuse the sale information without being rewarded in advance. More importantly, it gives the seller a nearly optimal revenue.

\section{Fixed-price Mechanism without Diffusion}
\label{without}
As we mentioned in Section 1, fixed-price mechanisms are widely used due to their simplicity. However, like many other traditional mechanisms such as the well-known Vickrey-Clarke-Groves (VCG) mechanism, the seller only sells the item among the people she can communicate with~\cite{vickrey1961counterspeculation}.

Given the model defined in Section 2, we can give a brief description of the fixed-price mechanism for selling one item as follows.
\begin{framed}
	\noindent\textbf{Fixed-price Mechanism without Diffusion}\\ 
	\rule{\textwidth}{0.5pt}
	\begin{enumerate}
		\item The seller sets a fixed price $p \in [0,1]$ and informs it to all buyers. 
		\item The buyers whose valuations are greater than the fixed price $p$ claim to buy the item. 
		\item The seller selects one of them who claimed to buy (with random tie-breaking), gives the item to her, and receives the payment $p$ from the buyer.
		\item If no one claims to buy the item, the seller keeps it.
	\end{enumerate}
\end{framed} 

In order to get higher revenue, the seller can optimize the price of $p$ in the above mechanism. Theorem 1 shows that if the seller knows how many buyers participating in the sale, there is an optimal fixed price for her to maximize her revenue.
\begin{theorem}
	Under the fixed-price mechanism without diffusion, given the number of buyers $x$ in the market, the optimal fixed-price to maximize the seller's expected revenue is $p^{opt}=(\frac{1}{1+x})^\frac{1}{x}$, which is an increasing function of $x$.
\end{theorem}
\begin{proof}
	Since the valuations of all buyers on the item have the independent and identical uniform distribution $U[0,1]$, the probability that a buyer's valuation $v_i<p$ equals $P(v_i<p)=p$, and the probabilty that all of the buyers' valuations are smaller than $p$ equals $p^x$, which is also the probability that the item is not sold. So, the probability that the item can be sold is $1-p^x$ with the seller's revenue being $p$. Thus we get the seller's expected revenue
	\begin{equation*}
	E=(1-p^x) \cdot p
	\end{equation*}
	and its derivative with respect to $p$ is
	\begin{equation*}
	\frac{dE}{dp}=1-(x+1) \cdot p^x.
	\end{equation*}
	
	By letting $\frac{dE}{dp}=0$, we get the optimal fixed-price, 
	\begin{equation}
	p^{opt}=(\frac{1}{1+x})^{\frac{1}{x}}.
	\end{equation} 
	
	Hence, the seller's maximum expected revenue is
	\begin{equation}
	E_{base}=(1-(p^{opt})^x)\cdot p^{opt}=(\frac{x}{1+x})\cdot (\frac{1}{1+x})^{\frac{1}{x}}.
	\end{equation}
	
	Let $y=(\frac{1}{1+x})^{\frac{1}{x}}$, then $\ln y=-\frac{1}{x}\cdot \ln(1+x)$. Therefore $\frac{d(\ln y)}{dx}=\frac{(1+x)\cdot\ln(1+x)-x}{x^2\cdot(1+x)}$. By $(1+x)\cdot\ln(1+x)-x > 0 $, which is easy to prove, we have $\frac{d(\ln y)}{dx} > 0$. Having a positive derivative, $\ln y$ is an increasing function of $x$, and so is $y$. 
	\qed
\end{proof} 

Furthermore, since $\frac{x}{1+x}$ is an increasing function of $x$, we conclude that $E_{base}$ is an increasing function of $x$. That is more buyers participating in the fixed-price mechanism, more expected revenue for the seller.    

Now, we consider an extreme case of the network, where all buyers are the seller's neighbours. In this case, the maximum of the seller's expected revenue under the fixed-price mechanism is $E_{opt}=(\frac{1}{1+k})^\frac{1}{k}\cdot(\frac{k}{1+k})$, where $k$ is the total number of buyer in the network.

If $x$ is the number of neighbours of the seller $s$, then $E_{base}$ and $E_{opt}$ define the lower bound and the upper bound of the seller's expected revenue under our diffusion mechanism that we propose next.

\section{Fixed-price Diffusion Mechanism}
\label{FPDM}
In the fixed-price mechanism discussed in Section 3, the seller sells the item only among her neighbours $n_s$. However, the seller will obtain higher expected revenue if the number of participants increases. Therefore, in this section, we design a mechanism, called fixed-price diffusion mechanism, to incentivize the seller's neighbours to diffuse the sale information to their neighbours for the seller to increase the number of participants. 

Our mechanism will iron the conflict between the seller aiming to improve her revenue and the buyers unwilling to diffuse the information to bring more rivals. 

Before introducing our mechanism, we need some extra notations. Given the tree-like social network, let $d_j$ be the depth of buyer $j \in K$ and $PH_j$ be the set of all buyers except $j$ on the path from the seller to $j$. We assume that the seller has $x$ neighbours and let $n_s=\{1,\cdots,x\}$. In addition, we call a subtree rooted at each $i \in n_s$ a branch and denote the set of the buyers in branch $i$ with $X_i$, the cardinality of $X_i$ with $k_i$, and the number of the rest buyers in the network except branch $i$ with $k_{-i}(i=1,\cdots,x)$. Without loss of generality, we assume $k_1 \geq k_2 \geq \cdots \geq k_x$. As a final notation, $\alpha \in [0,1]$ is a factor, which is used to adapt the amount of the rewards in the proposed mechanism. 

Our mechanism is defined as follows. 
\begin{framed}
	\noindent\textbf{Fixed-price Diffusion Mechanism (FPDM)}\\ 
	\rule{\textwidth}{0.5pt}
	
	Given the buyers' action profile $a' \in \mathcal{F}((n_1,\cdots,n_k))$ and $\alpha \in [0,1]$, compute the number of buyers in each branch, for each branch $i$ $(i=1,\cdots,x)$, by substituting $x$ in Equation (1) with $k_{-i}$, fixed price $p_i^{opt}(a')=(\frac{1}{1+k_{-i}})^{\frac{1}{k_{-i}}}$ is defined for all the buyers in branch $i$ if they receive the item. 
	\begin{itemize}
		\item \textbf{Allocation:} 
		\begin{enumerate}
			\item Starting from $i=1$, while $i \leq x$, repeat the following:
			
			The seller informs all buyers in branch $i$ of the price $p_i^{opt}(a')$.
			\begin{enumerate}
				\item Those buyers in branch $i$ whose valuations are greater than or equal to $p_i^{opt}(a')$ claim to buy the item.
				\item If there are multiple claimers, the seller picks the one with the smallest depth. If there are multiple such claimers, select the one with the largest number of neighbours to allocate the item (with random tie-breaking). Allocation ends.
				\item Otherwise, set $i=i+1$. 
			\end{enumerate}	
			\item If the item is not allocated in Step 1, the seller keeps it.
		\end{enumerate}
		\item \textbf{Payment:} If the item is allocated to a buyer $j$ in branch $w$, then for each buyer $l \in K$, her payment is
		\begin{equation*}p_l (a')=		\begin{cases}
		p_w^{opt}(a')& \text{if $\pi_l(a')=1$}\\ 
		(p_{base}-p_w^{opt}(a'))\alpha (\frac{1}{2})^{d_l}& \text{if $\pi_l(a')=0$}\\
		& \text{and $l \in PH_{j}$}\\
		0&  \text{otherwise}
		\end{cases}
		\end{equation*}
		where $p_{base}=(\frac{1}{1+x})^\frac{1}{x}$ is the optimal fixed price to sell among $n_s$.
	\end{itemize}
\end{framed}

Obviously, $p_i^{opt}(a')$ is irrelevant to all the buyers in $X_i$. That is to say, $p_i^{opt}(a')$ is independent of buyers' action, which will be used in the proof of the IC of our mechanism.

Intuitively, if no buyers want to buy the item, then the seller keeps the item; Otherwise, the seller first selects the branch with the largest number of buyers to sell the item. Later in this branch, among those whose valuations are higher than the fixed-price, the seller allocates the item to the buyer who is closest to the seller and of the most neighbours (with random tie-breaking). The recipient of the item pays her payment to the seller. 

Unlike in the fixed-price mechanism defined in Section 3, here the payment varies with different branches.  

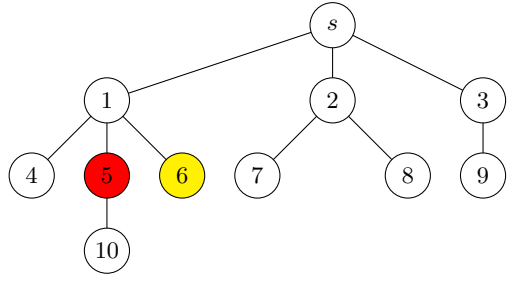
\begin{figure}
	\centering
	\begin{tikzpicture}
	\draw (1,1)--(2,2)--(5,3)--(7,2)--(7,1);
	\draw (2,2)--(2,1)--(2,0);
	\draw (2,2)--(3,1);
	\draw (5,3)--(5,2)--(4,1);
	\draw (5,2)--(6,1);
	\filldraw [fill=white] (1,1) circle (0.3);
	\filldraw [fill=white] (2,2) circle (0.3);
	\filldraw [fill=white] (2,0) circle (0.3);
	\filldraw [fill=red] (2,1) circle (0.3);
	\filldraw [fill=yellow] (3,1) circle (0.3);
	\filldraw [fill=white] (5,3) circle (0.3);
	\filldraw [fill=white] (5,2) circle (0.3);
	\filldraw [fill=white] (6,1) circle (0.3);
	\filldraw [fill=white] (7,2) circle (0.3);
	\filldraw [fill=white] (7,1) circle (0.3);
	\filldraw [fill=white] (4,1) circle (0.3);
	\node (P) at (5,3) {$s$};
	\node (P) at (2,2) {$1$};
	\node (P) at (5,2) {$2$};
	\node (P) at (7,2) {$3$};
	\node (P) at (1,1) {$4$};
	\node (P) at (2,1) {$5$};
	\node (P) at (3,1) {$6$};
	\node (P) at (4,1) {$7$};
	\node (P) at (6,1) {$8$};
	\node (P) at (7,1) {$9$};
	\node (P) at (2,0) {$10$};
	\end{tikzpicture}
	\label{proposed me}
	\caption{Illustrating the proposed mechanism. In the tree-structure social network, the numbers represent the buyers, and the seller $s$ is selling an item. The buyer 5 finally wins the item after competing with another buyer 6.}
\end{figure}

Figure~\ref{proposed me} illustrates the proposed mechanism with an example. The seller $s$ sells an item. There has three branches denoted by branch 1, 2 and 3. $X_1=\{1,4,5,6,10\}, X_2=\{2,7,8\}, X_3=\{3,9\}$. Then $p_1^{opt}(a')=0.699, p_2^{opt}(a')=0.743, p_3^{opt}(a')=0.760$. The valuations of all buyers from number 1 to number 10 are assumably given to be $v_1=0.6, v_2=0.7, v_3=0.7, v_4=0.5, v_5=0.8, v_6=0.9, v_7=0.3, v_8=0.4, v_9=0.1, v_{10}=0.5$ respectively. The seller first selects branch 1 since it has the largest number of buyers. Then in branch 1, according to $v_5 > p_1^{opt}(a'), v_6 >p_1^{opt}(a')$, we know only buyer 5 and 6 can buy the item. Since $d_5=d_6=2$, and since buyer 5 has one neighbour and buyer 6 has none, the seller allocates the item to buyer 5. We can work out that $PH_{5}=\{1\}$. At the same time, we get the seller's revenue generated by our mechanism is $p_1^{opt}(a')=0.699$ and $d_1=1$. Following the fixed-price mechanism defined in Section 3, we know that buyer 2 receives the item. Therefore $p_{base}=(\frac{1}{4})^\frac{1}{3}=0.630$. We let $\alpha=0.1$. The payments of buyer 1 and 5 are $-0.00345$ and $0.699$ respectively. So the utility of buyer 1 equals $ (p_1^{opt}(a')-p_{base})\alpha (\frac{1}{2})^{d_1}=0.00345$, the utility of buyer 5 equals $v_5-p_1^{opt}(a')=0.8-0.699=0.101$, and the other utilities equal to zero.

\section{Properties of FPDM}
\label{pro}
In this section, we first show that the seller's expected revenue generated by our mechanism is greater than that of the fixed-price mechanism without diffusion and approaches the expected revenue of the extreme optimal case defined in Section 3. Second, we prove that our fixed-price diffusion mechanism is individually rational. At last, we prove that our mechanism is incentive compatible. Therefore, the seller is willing to apply our mechanism.    
\begin{theorem}	
	The seller's expected revenue generated by our mechanism is greater than that of the fixed-price mechanism without diffusion; when the number of buyers in the social network is infinity and $\alpha \rightarrow 0$, the former is at least $50\%$ of the extreme optimal case defined in Section 3. This bound is independent of the structure of networks and is only related to the number of buyers in the network.
\end{theorem}
\begin{proof}
	We first prove the first part of the theorem. According to the payment policy of our mechanism, given the buyers' action profile $a' \in \mathcal{F}((n_1,\cdots,n_k))$, the seller's expected revenue generated by our mechanism equals
	\begin{equation*}
	\begin{split}
	E_{FPDM}=&p_1^{opt}(a')(1-p_1^{opt}(a')^{k_1})+\\
	&p_2^{opt}(a')p_1^{opt}(a')^{k_1}(1-p_2^{opt}(a')^{k_2})+\cdots+\\
	&p_x^{opt}(a')p_1^{opt}(a')^{k_1}\cdots p_{x-1}^{opt}(a')^{k_{x-1}}(1-p_x^{opt}(a')^{k_x}).
	\end{split}
	\end{equation*}
	
	By Equation (2), the seller's expected revenue generated by the fixed-price mechanism without diffusion equals
	\begin{equation*}
	E_{base}=(\frac{x}{1+x})\cdot (\frac{1}{1+x})^{\frac{1}{x}}.
	\end{equation*}
	
	It is a tedious manipulation to compare $E_{FPDM}$ and $E_{base}$. Without loss of generality, we assume $x=2, k_1=3, k_2=2$. Then 
	\begin{flalign*}
	&\text{$E_{FPDM}=(\frac{1}{3})^\frac{1}{2}\cdot (1-(\frac{1}{3})^\frac{3}{2}) +(\frac{1}{4})^\frac{1}{3} \cdot (\frac{1}{3})^\frac{3}{2} \cdot (1-(\frac{1}{4})^\frac{2}{3}))$,}&
	\end{flalign*}
	\begin{flalign*}
	&\text{$E_{base}=\frac{2}{3}\cdot(\frac{1}{3})^\frac{1}{2}$.}&
	\end{flalign*}
	
	It is obvious that $(\frac{1}{3})^\frac{1}{2}\cdot (1-(\frac{1}{3})^\frac{3}{2}) > \frac{2}{3}\cdot(\frac{1}{3})^\frac{1}{2}$, so we have $E_{FPDM} > E_{base}$. 
	
	To prove the second part of the theorem, for simplicity, we further assume that only one of the seller's neighbours has descendants.Then, the seller's expected revenue under our mechanism is	
	\begin{equation*}
	\begin{split}
	E_{FPDM}=&[1-(\frac{1}{x})^\frac{k-x+1}{x-1}](\frac{1}{x})^\frac{1}{x-1}+\\
	&(\frac{1}{x})^\frac{k-x+1}{x-1}[1-(\frac{1}{k})^\frac{x-1}{k-1}](\frac{1}{k})^\frac{1}{k-1}.
	\end{split}
	\end{equation*}
	
	Again, following the Section 3, we know that  $E_{opt}=(\frac{1}{1+k})^\frac{1}{k}\cdot (\frac{k}{1+k})$.
	
	To compare $E_{FPDM}$ with $E_{opt}$, the curve of $\frac{E_{FPDM}}{E_{opt}}$ is depicted in Figure~\ref{fig:222222}, from which we see that $E_{FPDM}$ achieves at least $50\%$ of $E_{opt}$.
	\qed
\end{proof}

For a clearer understanding of the theorem, we consider two special cases. 

\begin{enumerate}
	
\begin{figure}[h]
	\centering
	\includegraphics[width=0.8\linewidth]{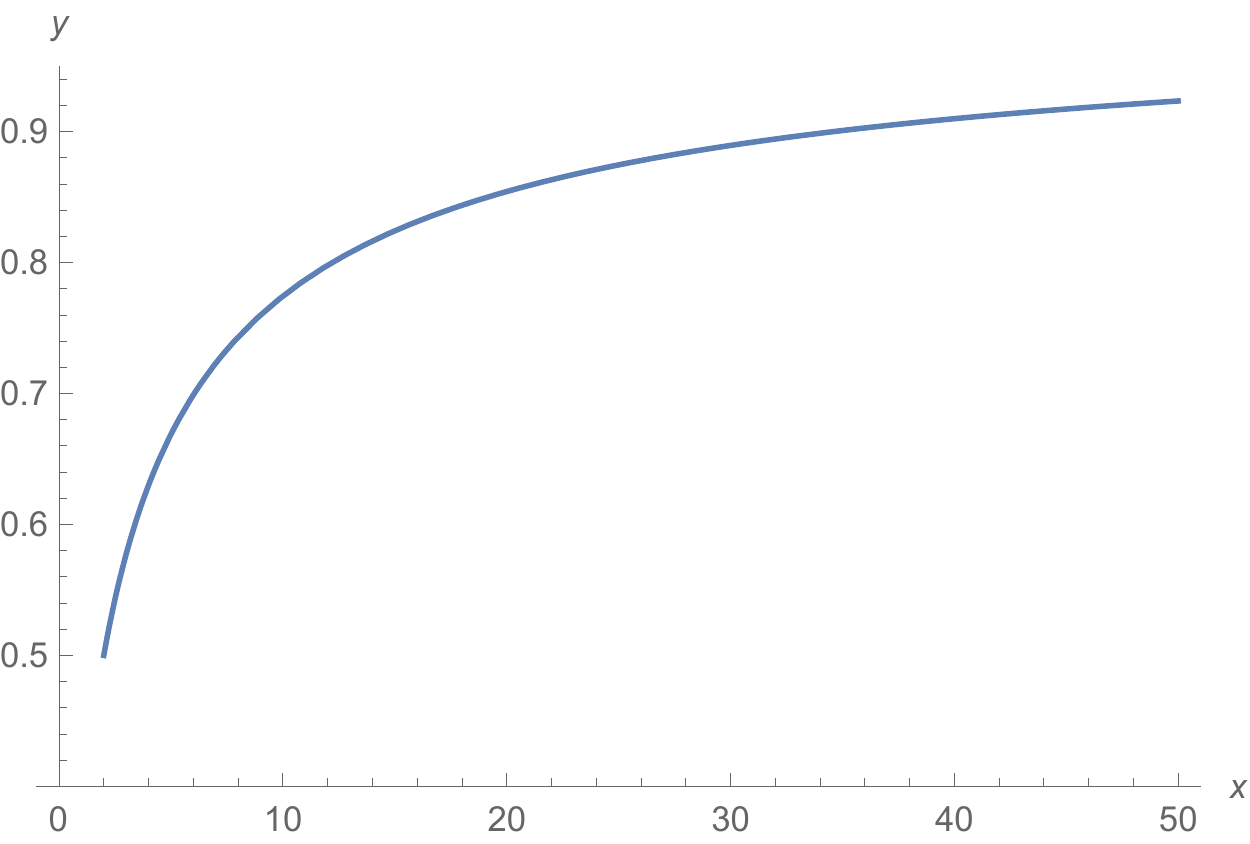}
	\caption{The curve of $\frac{E_{FPDM}}{E_{opt}}$ when the number of buyers is infinity.}
	\label{fig:222222}
\end{figure}
	
	\item The first case is one where the number of the seller's neighbours is a constant and where the number of buyers in the social network is infinity. For convenience, we let $x=5$ and fix $\alpha= 0.1$.
	
\begin{figure}
	\centering
	\includegraphics[width=0.8\linewidth]{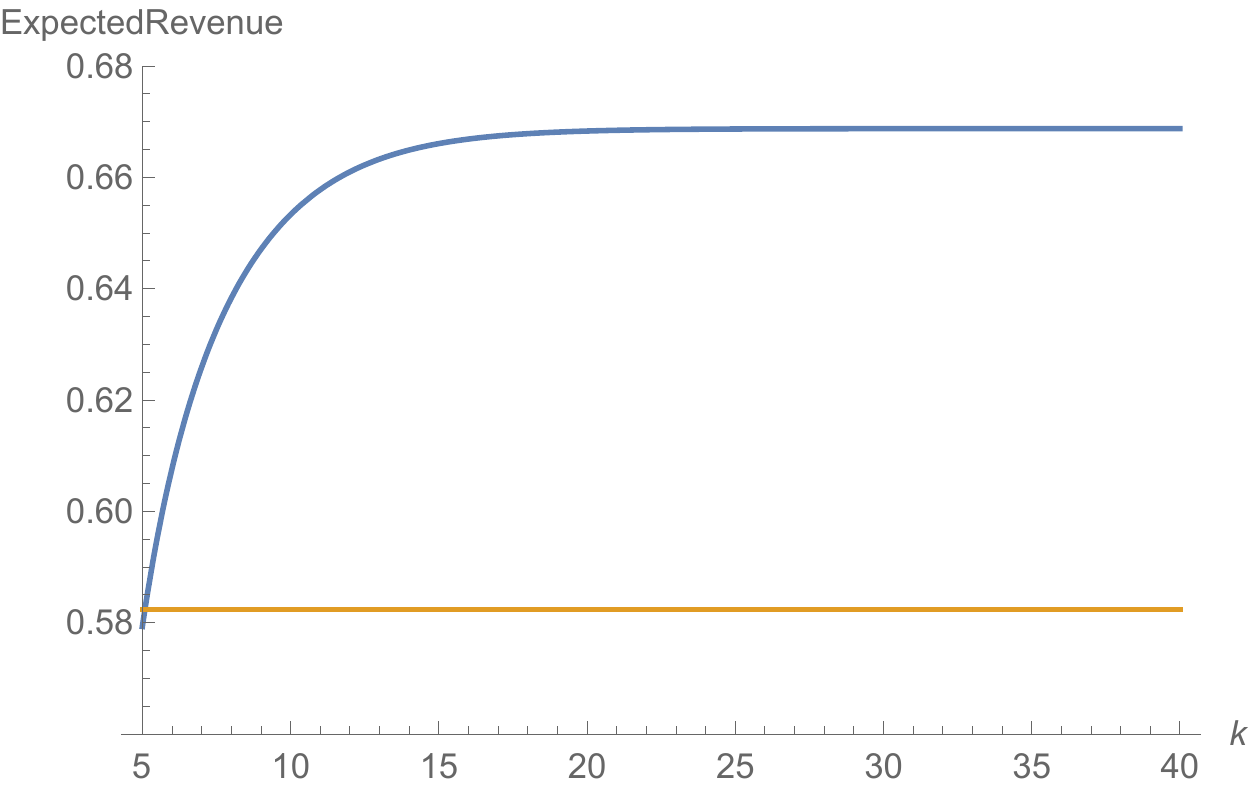}
	\caption{The curves of $E_{FPDM}$ and $E_{base}$ when $x=5$. The upper is the one of $E_{FPDM}$ and the lower the one of $E_{base}$.
	}
	\label{fig:333333}
\end{figure}
	
From the Figure~\ref{fig:333333}, we see that when $k \geq 5$, the seller's expected revenue generated by our mechanism is greater than that of the fixed-price mechanism without diffusion, i.e.,  $E_{FPDM}>E_{base}$.
	
	\item The second case is one where  $x=5$  and  $k \rightarrow +\infty$, we compute the following limitation
	\begin{equation*}
	\begin{split}
	&\lim\limits_{k\rightarrow +\infty}\frac{E_{FPDM}}{E_{opt}}\\
	=&\lim\limits_{k\rightarrow +\infty}\frac{	[1-(\frac{1}{x})^\frac{k-x+1}{x-1}]\frac{1}{x}^\frac{1}{x-1}+\frac{1}{x}^\frac{k-x+1}{x-1}[1-(\frac{1}{k})^\frac{x-1}{k-1}](\frac{1}{k})^\frac{1}{k-1}}{(\frac{1}{1+k})^\frac{1}{k}(\frac{k}{1+k})}\\
	=&(\frac{1}{5})^\frac{1}{4}\approx 67\%.
	\end{split}	
	\end{equation*}
	
\end{enumerate}

We show the curve of $Ratio=\frac{E_{FPDM}}{E_{opt}}$ in this case in Figure~\ref{fig:111111}, from which we see that $E_{FPDM}$ is more than $66\%$ of $E_{opt}$. 

From the above two cases, it can be induced that the seller's expected revenue generated by our mechanism is not only greater than that of the fixed-price mechanism without diffusion but also can achieve at least $66\%$ of the extreme optimal case defined in Section 3.

\begin{figure}
	\centering
	\includegraphics[width=0.8\linewidth]{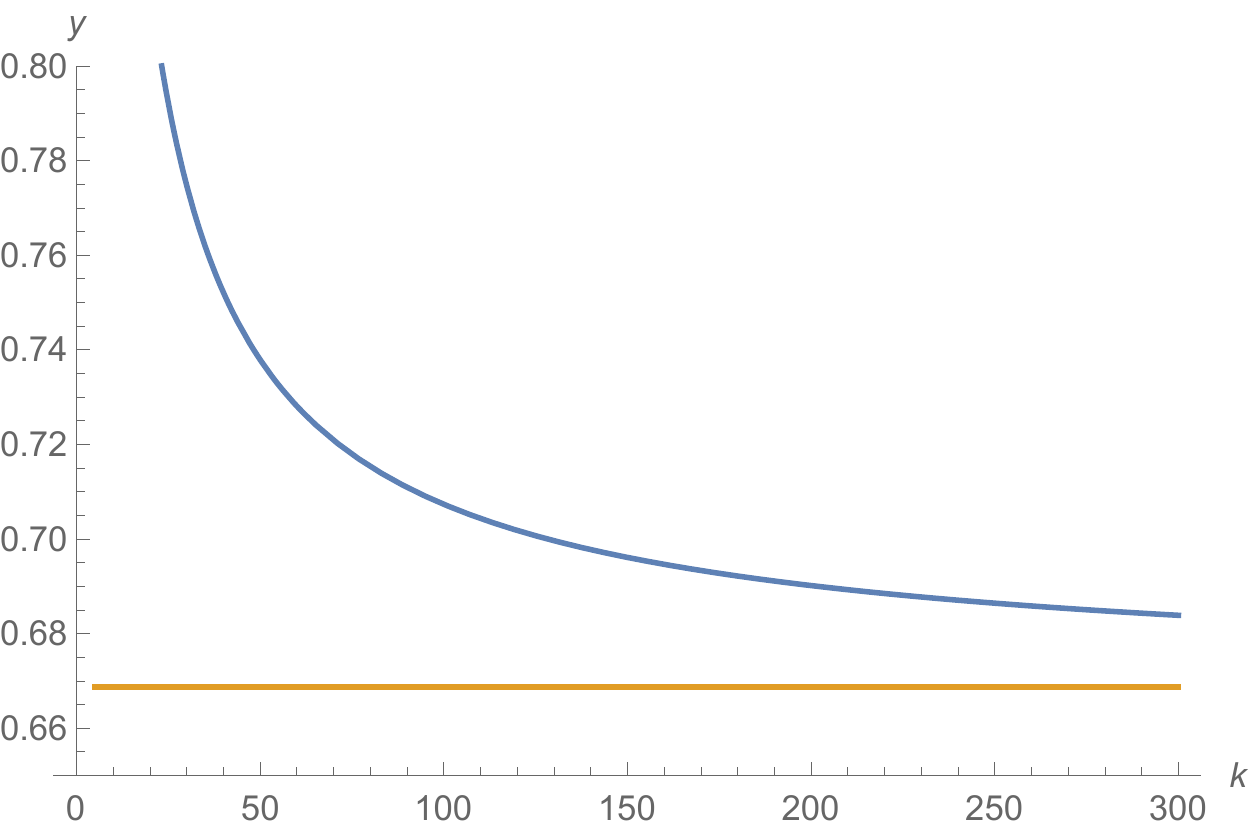}
	\caption{The curve of $\frac{E_{FPDM}}{E_{opt}}$, which approaches the horizontal line $Ratio\approx0.67$.}
	\label{fig:111111}
\end{figure}

Next, we investigate the individual rationality of the proposed mechanism.

\begin{theorem}
	The fixed-price diffusion mechanism is individually rational.
\end{theorem}
\begin{proof}
	Given the buyers' action profile $a' \in \mathcal{F}((n_1,\cdots,n_k))$, we only need to show that for all $i \in K$, it holds that 
	\begin{equation*}
	u_i(a',(\pi^{FPDM},p^{FPDM})) \geq 0
	\end{equation*}
	for all $a_i'=n_i'$.
	
	From the definition of FPDM, for any buyer $i \in K$, without loss of generality, we assume that $i \in X_j(j=1,\cdots,x)$. Then we have
	\begin{align*}
	&u_i(a',(\pi^{FPDM},p^{FPDM}))=0 \notag \\
	&u_i(a',(\pi^{FPDM},p^{FPDM}))=v_i-p_j^{opt}(a') \notag \\
	&u_i(a',(\pi^{FPDM},p^{FPDM}))=(p_j^{opt}(a')-p_{base})\alpha (\frac{1}{2})^{d_i}
	\end{align*}
	
	According to the allocation policy of the mechanism, we have $v_i > p_j^{opt}(a')$, so
	\begin{equation*}
	u_i(a',(\pi^{FPDM},p^{FPDM}))=v_i-p_j^{opt}(a') >0.
	\end{equation*}
	
	By the Theorem 2, we know that $p_j^{opt}(a') > p_{base}$. Then following the definition of $\alpha$ and $d_i$, we get
	\begin{equation*}
	u_i(a',(\pi^{FPDM},p^{FPDM}))=(p_j^{opt}(a')-p_{base})\alpha (\frac{1}{2})^{d_i} >0.
	\end{equation*}
	
	Therefore, we have $u_i(a',(\pi^{FPDM},p^{FPDM})) \geq 0$.
	\qed
\end{proof}

Now we discuss the incentive compatibility of the fixed-price diffusion mechanism. Based on the model defined in Section 2, we only need to analyze the diffusion action of each buyer. We obtain the following theorem by proving that for all buyers, diffusing the information to all their neighbours (i.e., $n_i'=n_i$) maximizes their utilities, which means the mechanism incentivizes them to do so.
\begin{theorem}
	The fixed-price diffusion mechanism is incentive compatible.
\end{theorem}
\begin{proof}
	Given the buyers' action profile $a' \in \mathcal{F}((n_1,\cdots,n_k))$, to prove that FPDM is incentive compatible, we need to show that for each buyer $i \in K$ such that $a_i' \neq nil$, diffusing the sale information to all $i$'s neighbours $n_i$ maximises $i$'s utility.
	We first classify all the buyers into three categories, then for buyers in each category prove the theorem is true.
	\begin{enumerate}
		\item The buyer who receives the item, denoted by $w$. 
		\item All buyers on the path from the seller to $w$.
		\item All buyers who are not included in the above two categories. 
	\end{enumerate} 
	
	\textbf{Category} (1): Without loss of generality, we assume $w \in X_j$. Then when reporting type truthfully, $w$'s utility is $u_w(a',(\pi^{FPDM},p^{FPDM}))=v_w-p_j^{opt}(a')$, which is independent of her diffusion action. Now when $w$ misreports her type, then her utility is also $v_w-p_j^{opt}(a')$. Therefore, the utility is not better than reporting type truthfully. Therefore, reporting type truthfully maximizes $w$'s utility.
	
	\textbf{Category} (2): According to the notation introduced in Section 4, for each buyer $i$ in this category, we have  $i \in PH_{w}$, where $w$ is the buyer who recieves the item. Without loss of generality, we assume $w \in X_j$. When reporting type truthfully, the buyer $i$'s utility is
	$u_i(a', (\pi^{FPDM},p^{FPDM}))=(p_j^{opt}(a')-p_{base})\alpha(\frac{1}{2})^{d_i}$. Since that $p_j^{opt}(a')=(\frac{1}{1+k_{-j}})^{\frac{1}{k_{-j}}}$, $k_{-j}$ is independent of $i$'s action, and $p_{base}$ and $d_i$ are irrevalent to $i$'s action, we know that $i$'s utility is independent of her action. When misreporting her type, then her utility does not change. Therefore it is not worth misreporting type.
	
	\textbf{Category} (3): For each buyer $i$ in this category, her utility is 0. If $i$ belongs to the branch which the winner belongs to, diffusing the sale information to all her neighbours makes $i$ or her descendants have the chance to become the winner, which increases $i$'s utility. Otherwise, $i$'s diffusing the sale information to all her neighbours increases the number of buyers in the branch which $i$ belongs to, which creates the chance that her branch  wins the item, and then $i$ will have the chance to become the winner or to stand in the path from the seller to the winner, which increases $i$'s utility. Therefore, diffusing the sale information to all her neighbours maximizes $i$'s utility.
	
	Putting the above together, we know that the fixed-price diffusion mechanism is incentive compatible.
	\qed
\end{proof}

\section{Conclusions}
\label{con}
This paper designed a mechanism called fixed-price diffusion mechanism (FPDM) for a seller to sell an item via a tree-structure social network. The mechanism helps the seller to attract more buyers by incentivizing buyers to invite each other on the network. Under this mechanism, the seller's revenue improvement is guaranteed, which is at least 1/2 of the optimal revenue the seller can get with the fixed-price mechanism. More importantly, the mechanism advances other existing solutions in the sense that it is fix-priced and does not require buyers' valuations to make the decision. 

We have only considered tree-structure social networks in this paper. It would be very interesting to look at the problem in a general network structure. We have already had an idea to tackle this problem. Roughly speaking, we transform the graph structure social network into one of tree structure by randomly deleting some edges, which we will prove to maintain the properties of our mechanism. Another interesting topic would be to consider hiding the social network structure from the seller as well. Lastly, there might exist other fixed-price mechanisms to further improve the seller's revenue.

\bibliographystyle{splncs04}
\bibliography{samplepaper}
\end{document}